\documentclass[review]{elsarticle}

\usepackage{lineno,hyperref}
\usepackage{enumerate}
\usepackage{fullpage}
\usepackage{setspace}
\usepackage{hyperref}
\usepackage[english]{babel}
\usepackage{hhline}
\usepackage{graphicx}
\usepackage{epstopdf}
\usepackage{amsfonts}
\usepackage{rotating}
\usepackage{amsmath}
\usepackage{amssymb}
\usepackage{mathrsfs}
\usepackage{bm}
\usepackage{color}
\usepackage{amsthm}
\usepackage{cases}
\usepackage{babel}
\usepackage{xspace}
\usepackage{psfrag}
\usepackage{multirow}
\usepackage{caption}
\usepackage[normalem]{ulem}
\usepackage{subfig}
\modulolinenumbers[250]

\usepackage[textwidth=6in]{geometry}

\makeatletter
\def\ps@pprintTitle{%
 \let\@oddhead\@empty
 \let\@evenhead\@empty
 \def\@oddfoot{}%
 \let\@evenfoot\@oddfoot}
\makeatother

\journal{Operations Research Letters}

\DeclareMathOperator*{\st}{s.t.}

\newtheorem{lemma}{Lemma}[section]

\newtheorem{prop}{Proposition}[section]
\newtheorem{thm}{Theorem}[section]
\newtheorem{rem}{Remark}[section]
\newtheorem{cor}{Corollary}[section]

\newtheorem{defi}{Definition}[section]

\hypersetup{
    colorlinks=true,
    citecolor=black,
    filecolor=black,
    linkcolor=black,
    urlcolor=black
}

\allowdisplaybreaks

\begin{document}

\begin{frontmatter}

\title{Multi-Period Portfolio Optimization: \\ {\Large Translation of Autocorrelation Risk to Excess Variance}}

\author[umich]{Byung-Geun Choi}
\author[epfl]{Napat Rujeerapaiboon}
\author[umich]{Ruiwei Jiang}

\address[umich]{Department of Industrial \& Operations Engineering, University of Michigan}
\address[epfl]{Risk Analytics and Optimization Chair, $\acute{\text{E}}$cole Polytechnique F$\acute{\text{e}}$d$\acute{\text{e}}$rale de Lausanne, Switzerland}

\begin{abstract}
Growth-optimal portfolios are guaranteed to accumulate higher wealth than any other investment strategy in the long run. However, they tend to be risky in the short term. For serially uncorrelated markets, similar portfolios with more robust guarantees have been recently proposed. This paper extends these robust portfolios by accommodating non-zero autocorrelations that may reflect investors' beliefs about market movements. Moreover, we prove that the risk incurred by such autocorrelations can be absorbed by modifying the covariance matrix of asset returns.
\end{abstract}


\end{frontmatter}

\linenumbers

\section{Introduction}
\label{sec:intro}
In this paper, we consider a dynamic portfolio optimization problem where investors face a challenge of how to allocate their wealth over a set of available assets to maximize their terminal wealth. By optimizing expected log-utility over a single investment period, the obtained portfolio, referred to as the growth-optimal portfolio, is shown to be optimal with respect to several interesting objectives in a classical stochastic setting. For example, \cite{kelly56} and \cite{breiman61} independently demonstrate that the growth-optimal portfolio will eventually accumulate more wealth than any other causal investment strategy with probability 1 in the long run. Moreover, it also minimizes the expected time required to reach a specified wealth target when the target is asymptotically large; see, e.g., \cite{breiman61,Algoet88}. For the readers interested in the history and the properties of the growth-optimal portfolio, we refer to \cite{Christensen05, MacLean10,Poundstone05}. Nonetheless, despite its theoretical appeals, there are many reasons why the practical relevance of the growth-optimal portfolio remains limited.

First, empirically the growth-optimal portfolio is shown to be highly volatile in the short run. Moreover, the calculation of the growth-optimal portfolio requires full and precise knowledge of the asset return distribution. In practice, this distribution is not available and has to be estimated from sparse empirical data. Therefore, the growth-optimal portfolio is prone to statistical estimation errors. \cite{Rujeerapaiboon16} extends the growth-optimal portfolio to more practical settings by proposing a fixed-mix investment strategy that offers a similar performance guarantee as the classical growth-optimal portfolio but for a finite investment horizon. Moreover, the proposed performance guarantee is not distribution-specific but remains valid for any asset return distribution within the prescribed ambiguity set of serially uncorrelated distributions.

Our contribution in this paper is to extend the results in \cite{Rujeerapaiboon16} to the case with non-zero autocorrelations (also known as, serial correlations). These autocorrelations can be used to incorporate beliefs of the investors about market movements as well as seasonality in asset returns; see e.g. \cite{Tinic84,Lo90}. Moreover, we prove that these autocorrelations can be absorbed in the covariance matrix underlying the asset return distribution. Finally, we remark that all of the discussed \emph{dynamic} investment strategies, namely the classical growth-optimal portfolio, the robust growth-optimal portfolio (see \cite{Rujeerapaiboon16}), and the extended robust growth-optimal portfolio (proposed in this paper), share similar computational advantage, in the sense that, all of them can be obtained with relative ease by solving \emph{static} optimization problems.

{\color{black}The rest of the paper is structured as follows. In Section~\ref{sec:dist} we explain how we model the distributional ambiguity in financial markets, and in Section~\ref{sec:riskmeasure} we define the risk measure, namely worst-case growth rate, to assess the performance of each individual portfolio. An analytical formula for the worst-case growth rate is then derived in Section~\ref{sec:wcgrowthrate}. A couple of numerical experiments are also given in this section. Finally, we provide an approximate worst-case growth rate for a more general probabilistic setting in Section~\ref{sec:extension}, and Section~\ref{sec:conclusion} concludes.
}

\textbf{Notation.} We denote the space of symmetric matrices in $\mathbb{R}^{n \times n}$ by $\mathbb S^n$. For any symmetric matrices $\mathbf X$ and $\mathbf Y$ with the same dimension, we denote their trace scalar product by $\left< \mathbf X, \mathbf Y \right>$. Moreover, for a positive semidefinite $\mathbf X \in \mathbb S^n$, we define $\mathbf X^{1/2}$ as its principle square root. We also define $\bm 1$ as a column vector of ones and $\mathbb I$ as an identity matrix. Their dimensions should be clear from the surrounding context. Random variables are represented by symbols with tildes. 
We denote by $\mathcal P_0^n$ the set of all probability distributions $\mathbb P$ on $\mathbb R^n$, and we represent by $\mathbb E_{\mathbb P}(\cdot)$ and $\mathbb{COV}_{\mathbb P}(\cdot,\cdot)$ the expectation and the covariance of the input random parameter(s) with respect to the probability distribution $\mathbb P$. Throughout the paper, we assume that the investment horizon is given by $\mathcal T = \lbrace 1, 2, \hdots, T \rbrace$ and the asset universe is given by $\mathcal N = \lbrace 1, 2, \hdots, N \rbrace$. Moreover, for a complex number $c$, we denote its real part by $\text{Re}(c)$. Finally, we define $(t)_T$ as the residue of $t$ modulo $T$. Note that for any $t \in \mathbb Z$, $(t)_T$ takes a value from $\lbrace 0, \hdots, T-1 \rbrace$.

\section{Distributional setting: stationary means, variances, and autocorrelations} \label{sec:dist}
In this section, we describe the setting of the probability distributions $\mathbb{P}$ of the asset returns $[\tilde{\bm r}_t]_{t=1}^T$. We define the ambiguity set
\begin{equation*}
\begin{aligned}
	\mathcal{P} = 
	\left \lbrace 
		\mathbb{P} \in \mathcal{P}_0^{NT}~: \begin{array}{l}
			\mathbb{E}_\mathbb{P} \left( \tilde{r}_{t,i} \right) = \mu_i
				\quad \forall i \in \mathcal{I} 
				\quad \forall t \in \mathcal{T} \\
			\mathbb{COV}_\mathbb{P} \left( \tilde{r}_{s,i}, \tilde{r}_{t,j} \right) = 
				\rho_{(t-s)_T} \sigma_{i,j}
				\quad \forall i,j \in \mathcal{I} 
				\quad \forall s,t \in \mathcal{T}
		\end{array}
	\right \rbrace,
\end{aligned}
\end{equation*}
where $\bm\mu = [\mu_i]_{i \in \mathcal{N}} \in \mathbb{R}^N$ stands for the vector of expected asset returns and $\bm\Sigma = [\sigma_{i,j}]_{i,j \in \mathcal{N}} \in \mathbb{S}^{N}$ stands for the covariance matrix of asset returns. Both $\bm\mu$ and $\bm\Sigma$ are assumed to be stationary, i.e., they remain unchanged over time. This choice of ambiguity set $\mathcal{P}$ is nicely motivated in \cite{Roy52}. In contrast to \cite{Rujeerapaiboon16} which assumes uncorrelatedness between asset returns at different trading periods $s$ and $t$, we allow them to be correlated with correlation $\rho_{(t-s)_T}$. 

Next, we consider a fixed-mix strategy (see \cite{Rujeerapaiboon16}) generated by a portfolio $\bm w \in \mathcal{W} \subset \mathbb{R}^N_+$, where $\mathcal{W}$ is a convex polyhedral of allowable portfolios. That is, we assume that portfolio weights revert to $\bm{w}$ at every rebalancing date $t \in \mathcal{T}$. 
In order to avoid clutter, we denote the portfolio return during the trading period $t$ by $\tilde\eta_t = \bm{w}^\intercal\bm{\tilde r}_t$. It is clear that, under the distribution $\mathbb{P}$ of asset returns, we have that
\begin{equation*}
\begin{aligned}
	\mathbb{E}_\mathbb{P} (\tilde \eta_t) = \bm w^\intercal \bm\mu 
	\quad \text{and} \quad
	\mathbb{COV}_\mathbb{P} (\tilde \eta_s, \tilde \eta_t) = \rho_{(t-s)_T} \bm w^\intercal \bm\Sigma \bm w
	\quad \forall s,t \in \mathcal{T}.
\end{aligned}
\end{equation*}
Equivalently put, the mapping $\tilde\eta_t = \bm{w}^\intercal\bm{\tilde r}_t$ projects $\mathcal{P}$ to an ambiguity set $\mathcal{P}(\bm{w})$ defined as
\begin{equation*}
\begin{aligned}
	\mathcal{P}(\bm w) = 
	\left \lbrace 
		\mathbb{P} \in \mathcal{P}_0^{T}~: \begin{array}{l}
			\mathbb{E}_\mathbb{P} (\tilde \eta_t) = \bm w^\intercal \bm\mu
				\quad \forall t \in \mathcal{T} \\
			\mathbb{COV}_\mathbb{P} (\tilde \eta_s, \tilde \eta_t) = \rho_{(t-s)_T} \bm w^\intercal \bm\Sigma \bm w
				\quad \forall s,t \in \mathcal{T}
		\end{array}
	\right \rbrace.
\end{aligned}
\end{equation*}
The projection property of the ambiguity set (see, for example, \cite[Theorem~1]{Yu09} and \cite[Proposition~2]{Rujeerapaiboon16}) further asserts the existence of the inverse mapping from $\mathcal P (\bm w)$ to $\mathcal P$. Hence, from the next section onward, we will study the performance of an arbitrary portfolio $\bm{w}$ using the projected ambiguity set $\mathcal{P}(\bm w)$ instead of the original ambiguity set $\mathcal{P}$, since maneuvering $\mathcal{P}(\bm w)$ often leads to an optimization problem with a smaller dimension.

\section{Portfolio performance measure: worst-case growth rate} \label{sec:riskmeasure}
In a friction-less  market, a fixed portfolio $\bm w$ repeatedly invested over the investment horizon $\mathcal{T}$ leads to a total return of
\begin{equation*}
	\prod_{t \in \mathcal T} ( 1 + \bm{w}^\intercal \tilde{\bm{r}}_t) = \text{exp} \left( \sum_{t \in \mathcal T} \log ( 1 + \bm{w}^\intercal \tilde{\bm{r}}_t) \right),
\end{equation*}
which is a random amount. An intuitive performance measure for this portfolio would thus be an expectation of its logarithmic terminal wealth $\mathbb{E} \left( \sum_{t \in \mathcal T}\left( \log ( 1 + \bm{w}^\intercal \tilde{\bm{r}}_t) \right) \right)$. 
A portfolio that maximizes such utility function is referred to as a growth-optimal portfolio. This portfolio exhibits many intriguing asymptotic properties, and some of them are discussed in Section \ref{sec:intro}. Moreover, if the asset return distribution is serially independent and identically distributed, then the growth-optimal portfolio can be obtained by solving a static optimization problem $\max_{\bm w \in \mathcal W}\mathbb{E} \left( \log ( 1 + \bm{w}^\intercal \tilde{\bm{r}}_1) \right)$. However for a finite $T$ (especially when $T$ is small), the expectation criterion becomes risky (as the central limit theorem fails) and accordingly \cite{Rujeerapaiboon16} proposes to use a quantile criterion instead of the expectation criterion. Precisely speaking, \cite{Rujeerapaiboon16} employs recent advances in distributionally robust optimization (see \cite{Zymler10}) to determine the \emph{worst-case growth rate} by solving the following optimization problem\footnote{Note that $\bm{w}^\intercal\tilde{\bm{r}}_t - \frac{1}{2}\left(\bm{w}^\intercal \tilde{\bm{r}}_t \right)^2$ is a second-order Taylor approximation of $\log(1 + \bm{w}^\intercal \tilde{\bm{r}}_t)$ around $\bm{w}^\intercal \tilde{\bm{r}}_t = 0$. The approximation becomes more accurate as the rebalancing frequency increases.}
\begin{equation*}
\begin{aligned}
	\mathcal{G}_{\epsilon}(\bm{w}) 
	&= \max_\gamma \left\lbrace \gamma :
		\mathbb{P} \left( \frac{1}{T} \sum_{t \in \mathcal T} \left( \bm{w}^\intercal \tilde{\bm{r}}_t - \frac{1}{2}\left(\bm{w}^\intercal \tilde{\bm{r}}_t \right)^2 \right) \geq \gamma \right) \geq 1 - \epsilon \quad \forall \mathbb{P} \in \mathcal{P}
	\right\rbrace \\
	&= \max_\gamma \left\lbrace \gamma :
		\mathbb{P} \left(  \frac{1}{T} \sum_{t \in \mathcal T} \left( \tilde\eta_t - \frac{1}{2}\tilde\eta_t^2 \right) \geq \gamma \right) \geq 1 - \epsilon \quad \forall \mathbb{P} \in \mathcal{P}(\bm w)
	\right\rbrace.
\end{aligned}	 
\end{equation*}
where the ambiguity set in \cite{Rujeerapaiboon16} is the restriction of ours where $\rho_t = 0$ for every $t = 1, \hdots, T-1$. The uncorrelatedness assumption allows \cite{Rujeerapaiboon16} to solve this distributionally robust program efficiently because of the inherent temporal symmetry. Our work relaxes this assumption in order to accommodate investors' beliefs and market seasonality. In particular, we show that, despite the fact that the temporal symmetry breaks down, we can still derive an analytical expression of $\mathcal{G}_{\epsilon}(\bm{w})$ by using knowledge from linear algebra of circulant matrices; see e.g.~\cite{Gray01}. We highlight that even though the relaxation does not change the problem greatly, it still requires us to develop new mathematical techniques to accommodate these changes.

\section{Derivation of worst-case growth rates}
\label{sec:wcgrowthrate}
By using the semidefinite program reformulation for distributionally robust quadratic chance constraints provided in \cite{Zymler10}, we can rewrite $\mathcal{G}_{\epsilon}(\bm{w})$ as
\begin{equation}
\label{opt:projected_sdp}
\begin{aligned}
	\mathcal{G}_{\epsilon}(\bm{w}) ~=~
	&\max && \gamma \\ 
	&\st &&  \mathbf{M} \in \mathbb{S}^{T+1},~\beta \in \mathbb{R},~\gamma \in \mathbb{R} \\
	& && \beta + \tfrac{1}{\epsilon}\left<\bm{\Omega}(\bm{w}), \mathbf{M} \right> \leq 0, \quad \mathbf{M} \succeq \boldsymbol{0}  \\
	& && \mathbf{M} - \left[ \begin{array}{cc}
							\frac{1}{2}\mathbb{I} & -\frac{1}{2}\bm{1} \\
							-\frac{1}{2}\bm{1}^{\intercal} & \gamma T - \beta
					   \end{array} \right] \succeq \bm{0},
\end{aligned}
\end{equation} 
where $\bm{\Omega}(\bm{w}) \succ \bm{0}$ is the projected second-order moment matrix for a sequence of portfolio returns $[\tilde{\eta}_t]_{t \in \mathcal T}$ generated by any distribution residing in $\mathcal{P} (\bm w)$, i.e., 
\begin{equation*}
	\boldsymbol{\Omega}(\bm{w}) = 
	\left[ \begin{array}{c|c}
		\bm{w}^\intercal \bm{\Sigma w} \cdot \mathbf{P} + (\bm{w}^\intercal \bm{\mu})^2 \cdot \bm{11}^\intercal &
        \bm{w}^\intercal \bm{\mu} \cdot \bm{1} \\
        \hline 
        \bm{w}^\intercal \bm{\mu} \cdot \bm{1}^\intercal & 
        1
	\end{array} \right]
\end{equation*}
{\color{black} and $\mathbf{P}$ is the autocorrelation matrix defined as
\begin{equation}
\label{eq:corrmat}
\begin{aligned}
	\mathbf{P} = \left[ \begin{array}{cccc}
		\rho_0 & \rho_1 & \hdots & \rho_{T-1} \\
		\rho_{T-1} & \rho_0 & \hdots & \rho_{T-2} \\
		\vdots & \vdots & \ddots & \vdots \\
		\rho_1 & \rho_2 & \hdots & \rho_0
	\end{array} \right].
\end{aligned}
\end{equation}
Note that for $\mathbf{P}$ to be a proper and non-degenerate autocorrelation matrix, we require that: $\rho_0 = 1$, $\rho_t = \rho_{T-t}$ for $t = 1, \hdots, T-1$ to ensure that $\mathbf{P}$ is symmetric, and $\mathbf{P} \succ \bm{0}$. Similarly, we also assume that $\bm\Sigma \succ \bm 0$ to eliminate degenerate cases. We highlight that these assumptions are non-restrictive and are almost always satisfied when there exists no risk-free asset in $\mathcal N$. We henceforth assume throughout the paper that they hold. 
}

Observe that the dimension of $\mathbf M$ scales with $T$. Therefore, directly solving this program for large $T$ is computationally prohibitive. Fortunately in our case, the upper left part of $\bm\Omega(\bm w)$ forms a circulant matrix. 
\begin{defi}[circulant matrix]
For $c_0, c_1, \hdots, c_{T-1} \in \mathbb{R}$, a circulant matrix $\mathbf{C} = \text{circ}(c_0, c_1, \hdots \allowbreak , c_{T-1}) \in \mathbb{R}^{T \times T}$ is defined as
\begin{equation*}
	\mathbf{C} = 
	\left[ \begin{array}{cccc}
		c_0 &
		c_1 &
		\cdots &
		c_{T-1} \\
		c_{T-1} &
		c_0 &
		\cdots &
		c_{T-2} \\
		\vdots &
		\vdots &
		\ddots &
		\vdots \\
		c_1 &
		c_2 &
		\cdots &
		c_0
	\end{array} \right].
\end{equation*}
\end{defi}
We circumvent the complexity issue in solving formulation \eqref{opt:projected_sdp} by exploiting the circulant property of matrix $\mathbf\Omega(\bm w)$. We demonstrate this result in Lemma 4.1 below.\begin{lemma}[circulant optimizer]
\label{lem:symmetry}
There exists an optimizer $(\mathbf{M}^\star, \beta^\star, \gamma^\star)$ of \eqref{opt:projected_sdp} where 
\begin{equation*}
	\mathbf{M}^\star = \left[ \begin{array}{c|c}
		\text{circ}(m_0, \hdots, m_{T-1}) & m_T \bm{1} \\
		\hline
		m_T \bm{1}^\intercal & m_{T+1}
	\end{array} \right]
\end{equation*}
for some $m_0, \hdots, m_{T+1} \in \mathbb{R}$.
\end{lemma}
\begin{proof}
For any permutation $\pi$ of the integers $\lbrace 1, 2, \hdots, T+1 \rbrace$, denote by $\mathbf{P}_\pi$ the permutation matrix which is defined through $[\mathbf{P}_\pi]_{i,j} = 1$ if $\pi(i) = j; = 0$ otherwise. Denote by $\Pi$ the set of permutations with the following properties.
\begin{enumerate}
	\item
	$\pi(T+1) = T+1$.
	\item
	For $d \in \lbrace +1, -1 \rbrace$, $\left( \pi(T) - \pi(1) \right)_T = \left( \pi(i) - \pi(i+1) \right)_T = (d)_T$~for~$i = 1, \hdots, T-1$.
\end{enumerate}
Note that $\vert \Pi \vert = 2T$ because $\pi(1)$ can be chosen freely from $1, \hdots, T$ and there are two possibilities for shifting direction $d$. Due to the circulant property of $\bm\Omega(w)$, it can be observed that $\mathbf{P}_\pi \bm\Omega(\bm w) \mathbf{P}_\pi^\intercal = \bm\Omega(\bm w)$ for any $\pi \in \Pi$. The proposition now follows from an argument parallel to that of \cite[Proposition~3]{Rujeerapaiboon16}.
\end{proof}
Without any loss of generality, Lemma \ref{lem:symmetry} allows us to restrict our attention to $\mathbf{M}$ of a specific form consisting of a circulant matrix, an extra column, and an extra row. Note that $\mathbf{M}$ is symmetric, implying that $m_t = m_{T-t}$ for $t = 1, \hdots, T-1$. Lemma \ref{lem:eigen} below shows that eigenvalues of any symmetric circulant matrix can be analytically determined.
\begin{lemma}[eigenvalues of circulant matrices]
\label{lem:eigen}
All eigenvalues of a symmetric circulant matrix $\mathbf{C} = \text{circ}(c_0, c_1, \hdots, c_{T-1})$ are $\sum_{t=0}^{T-1} c_t \cos \left( \frac{2\pi j t}{T}\right)$,~ $j = 0, \hdots, T-1$.
\end{lemma}
\begin{proof}
The eigenvalues of $\mathbf{C} = \text{circ}(c_0, c_1, \hdots, c_{T-1})$ are $\sum_{t=0}^{T-1} c_t \omega_j^t$ for $j = 0, \hdots, T-1$ where $\omega_j$ are $j$th roots of unity (see, for example, \cite[Chapter~3]{Gray01}). In addition, we know that all eigenvalues are real because $\mathbf{C}$ is symmetric. Dropping the imaginary parts in the expression of the eigenvalues, the claim follows.
\end{proof}

We are now ready to simplify \eqref{opt:projected_sdp} from a semidefinite program to a second-order cone program. To achieve this, we consider the first semidefinite constraint $\mathbf{M} \succeq \bm{0}$. Using Lemma \ref{lem:symmetry} and Lemma \ref{lem:eigen} together, we may rewrite this constraint as follows\footnote{The end result still holds when $m_{T+1} = 0$, which can be treated via case distinction. However, we omit this argument for the sake of brevity.}:
\begin{equation*}
\begin{aligned}
	\mathbf{M} \succeq \bm{0} \quad &\Longleftrightarrow \quad
	m_{T+1} \geq 0, \quad 
	\text{circ}(m_0, \hdots, m_{T-1}) \succeq \textstyle\frac{m_T^2}{m_{T+1}} \bm{11}^\intercal \\
	&\Longleftrightarrow \quad
	m_{T+1} \geq 0, \quad 	
	\text{circ}(m_0 - m_T^2/m_{T+1}, \hdots, m_{T-1} - m_T^2/m_{T+1}) \succeq \bm{0} \\
	&\Longleftrightarrow \quad
	m_{T+1} \geq 0, \quad 	
	\textstyle\sum_{t=0}^{T-1} \left( m_t - m_T^2/m_{T+1} \right)\cos \left( \frac{2\pi j t}{T}\right) \geq 0
	\quad j = 0, \hdots, T-1 \\
	&\Longleftrightarrow \quad
	m_{T+1} \geq 0, \quad 
	m_{T+1} \textstyle\sum_{t=0}^{T-1} m_t \geq Tm_T^2, \quad 
	\sum_{t=0}^{T-1} m_t \cos \left( \frac{2\pi j t}{T}\right) \geq 0
	\quad j = 1, \hdots, T-1,
\end{aligned}
\end{equation*}
where the first equivalence holds due to Schur complement, the third equivalence follows from Lemma \ref{lem:eigen}, and the last equivalence follows from 
\begin{equation}
\label{eq:sumroots}
	\sum_{t=0}^{T-1} \cos \left( \frac{2\pi j t}{T}\right) = \sum_{t=0}^{T-1}\text{Re}(\omega_j^t) = 	\text{Re} \left( \frac{1-\omega_j^T}{1-\omega_j} \right) = 0 \quad\text{for}~ j = 1, \hdots, T-1,
\end{equation}
where $\omega_j$ represent the $j$th roots of unity. The other semidefinite constraint in \eqref{opt:projected_sdp} can be reformulated in a similar manner, and thus we end up with the following reformulation of \eqref{opt:projected_sdp}.
\begin{subequations}
\label{opt:projected_socp}
\begin{alignat}{2}
	\mathcal{G}_{\epsilon}(\bm{w}) ~=~
	&\max~~ && \gamma \nonumber \\ 
	&\st &&  (m_0, \hdots, m_{T+1})^\intercal \in \mathbb{R}^{T+2},~\beta \in \mathbb{R},~\gamma \in \mathbb{R} \\
	& && m_{T+1} \geq 0, \quad
	     m_{T+1} - \gamma T + \beta \geq 0, \quad
	     m_t = m_{T-t} \quad t = 1, \hdots, T-1 \\
	& && m_{T+1} \textstyle\sum_{t=0}^{T-1} m_t \geq Tm_T^2, \quad
	     (m_{T+1} - \gamma T + \beta) \left( m_0 - \textstyle\frac{1}{2} + \sum_{t=1}^{T-1} m_t \right) \geq T(m_T + \textstyle\frac{1}{2})^2 \\
	& && m_0 - \textstyle\frac{1}{2} + \textstyle\sum_{t=1}^{T-1} m_t \cos \left( \frac{2\pi j t}{T}\right) \geq 0 \quad j = 1, \hdots, T-1 \label{eq:const-1}\\
    & && \epsilon\beta + T\textstyle\sum_{t=0}^{T-1} \left( \rho_t \bm w^\intercal \bm\Sigma \bm w + (\bm w^\intercal\bm\mu)^2\right) m_t + 2T\bm w^\intercal\bm\mu m_T + m_{T+1} \leq 0 \label{eq:const-2}
\end{alignat}
\end{subequations}
Note that all constraints in \eqref{opt:projected_socp} are either linear or hyperbolic, i.e., second-order cone representable. Hence, we have reformulated formulation \eqref{opt:projected_sdp} as a second-order cone program. Below, we derive an analytical solution for $\mathcal G_\epsilon(\bm w)$. Suppose that $x = (m_0, \ldots, m_{T+1}, \beta, \gamma)^{\top}$ is an optimal solution to formulation \eqref{opt:projected_socp}. Construct a new solution 
\begin{equation*}
	\textstyle
	\bm x' = \left( \underbrace{\left(\sum_{t=0}^{T-1} m_t - \frac{1}{2}\right)/T + \frac{1}{2}}_{m'_0}, \underbrace{\left(\sum_{t=0}^{T-1} m_t - \frac{1}{2}\right)/T}_{m'_1}, \hdots, \underbrace{\left(\sum_{t=0}^{T-1} m_t - \frac{1}{2}\right)/T}_{m'_{T-1}}, m_T, m_{T+1}, \beta, \gamma \right)^\intercal.
\end{equation*}
$\bm x'$ satisfies all constraints of \eqref{opt:projected_socp} except \eqref{eq:const-1}--\eqref{eq:const-2} because the transformation from $\bm x$ to $\bm x'$ preserves $\sum_{t=0}^{T-1} m_t$, i.e., $\sum_{t=0}^{T-1} m'_t = \sum_{t=0}^{T-1} m_t$. In the following, we argue that $\bm x'$ is indeed feasible in the view of constraints \eqref{eq:const-1}--\eqref{eq:const-2} as well. 

For constraint \eqref{eq:const-1}, $\bm x'$ is indeed feasible because $m'_0 - \frac{1}{2} = m'_1 = m'_2 = \hdots = m'_{T-1}$ and $\sum_{t=1}^{T-1} \cos \left( \frac{2\pi j t}{T}\right) = -1$, as previously pointed out in \eqref{eq:sumroots}.

For constraint \eqref{eq:const-2}, it is easier to look at the original version of this constraint in \eqref{opt:projected_sdp}, i.e., $\beta + \frac{1}{\epsilon} \left<\bm{\Omega}(\bm{w}), \mathbf{M} \right> \leq 0$. 
Note that, it is sufficient to show that
\begin{equation*}
	\textstyle
	\text{circ}(m'_0, m'_1, \hdots, m'_{T-1}) \preceq \text{circ}(m_0, m_1, \hdots, m_{T-1}). 
\end{equation*}
We let $m'$ denote the value shared by $m'_1, \hdots, m'_{T-1}$. Since the difference between two circulant matrices remains circulant, the above positive semidefinite constraint holds iff
\begin{equation*}
\begin{aligned}
	\text{circ}( m' + \textstyle\frac{1}{2}& - m_0, m' - m_1, \hdots, m' - m_{T-1} ) \preceq \bm{0} \\
	&\Longleftrightarrow ~
	m' + \textstyle\frac{1}{2} - m_0 + \textstyle\sum_{t=1}^{T-1} (m' - m_t)\omega_j^t \leq 0 \quad j = 0, \hdots, T-1 \\
	&\Longleftrightarrow ~
	m' + \textstyle\frac{1}{2} - m_0 + \textstyle\sum_{t=1}^{T-1} (m' - m_t)\omega_j^t \leq 0 \quad j = 1, \hdots, T-1 \\
	&\Longleftrightarrow ~
	\textstyle\frac{1}{2} + m' \textstyle\sum_{t=0}^{T-1} \omega_j^t \leq \sum_{t=0}^{T-1} m_t\omega_j^t \quad j = 1, \hdots, T-1 \\
	&\Longleftrightarrow ~
	\textstyle\frac{1}{2} + m' \left( \frac{1-\omega_j^T}{1-\omega_j} \right) \leq \sum_{t=0}^{T-1} m_t\omega_j^t \quad j = 1, \hdots, T-1 \\
	&\Longleftrightarrow ~
	\textstyle\frac{1}{2} \leq \sum_{t=0}^{T-1} m_t \text{Re}(\omega_j^t) \quad j = 1, \hdots, T-1 \\
	&\Longleftarrow ~
	 m_0 - \textstyle\frac{1}{2} + \textstyle\sum_{t=1}^{T-1} m_t \cos \left( \frac{2\pi j t}{T}\right) \geq 0 \quad j = 1, \hdots, T-1 \\
	&\Longleftarrow ~ 
	 \text{Feasibility of } \bm x ,
\end{aligned}
\end{equation*}
where the first equivalence is a consequence of Lemma \ref{lem:eigen} and the second equivalence excludes the case where $\omega = 1$ which trivially holds because the definition of $m'$ implies $Tm' + \frac{1}{2} = \sum_{t=0}^{T-1}m'_t = \sum_{t=0}^{T-1}m_t$.

Therefore, $\bm{x}'$ is feasible to \eqref{opt:projected_socp} and also optimal because it shares the same objective function value with the original optimal solution $\bm{x}$. It follows that, without loss of optimality, we can assume that $m_0 - \frac{1}{2} = m_1 = \cdots = m_{T-1}$ in \eqref{opt:projected_socp}. Finally, we obtain the analytical solution to \eqref{opt:projected_socp} in Theorem \ref{thm:wcvar} below.

\begin{thm}[worst-case growth rate]
\label{thm:wcvar}
If $1 - \bm{w}^\intercal\bm\mu > \sqrt{\frac{(1 + (T-1)\bar\rho)\epsilon}{(1-\epsilon) T}} \Vert \bm\Sigma^{1/2} \bm w \Vert$, then
\begin{equation*}
	\mathcal{G}_\epsilon(\bm w) = \frac{1}{2} \left( 1 - \left( 1 - \bm w^\intercal\bm\mu + \sqrt{\frac{(1-\epsilon)(1 + (T-1)\bar\rho)}{\epsilon T}} \Vert \bm{\Sigma}^{1/2}\bm w \Vert\right)^2 - \frac{T-1-(T-1)\bar\rho}{\epsilon T}\bm{w}^\intercal\bm\Sigma \bm w \right),
\end{equation*}	
where $\bar\rho$ is a constant defined as $(\sum_{t=1}^{T-1} \rho_t)/(T-1)$.
\end{thm}
\begin{proof}
The latest implication of the discussion prior the theorem suggests that there exists an optimal matrix $\mathbf{M}$ which is compound symmetric (see \cite[Definition~5]{Rujeerapaiboon16}). Therefore, \eqref{opt:projected_sdp} reduces to an optimization problem with two positive semidefinite constraints involving compound symmetric matrices. 
The claim thus follows by invoking \cite[Lemma~1]{Rujeerapaiboon16} through appropriate variable substitutions.
\end{proof}

Lemma~\ref{rem:rgop} below relates Theorem~\ref{thm:wcvar} with Theorem~2 in \cite{Rujeerapaiboon16}, whereas Corollary~\ref{cor:rgop} demonstrates how the portfolio with maximum worst-case growth rate  can be achieved. In the same spirit as \cite{Rujeerapaiboon16}, we refer to this portfolio as an \emph{extended robust growth-optimal portfolio}.

\begin{rem}
\label{rem:rgop}
When $\bar\rho = 0$, we recover the result from \cite[Theorem~2]{Rujeerapaiboon16} which provides a solution for the special case where $\rho_1 = \hdots = \rho_{T-1} = 0$. However, Theorem \ref{thm:wcvar} generalizes this result and implies that \cite[Theorem~2]{Rujeerapaiboon16} holds as long as $\sum_{t=1}^{T-1} \rho_t = 0$.
\end{rem}

\begin{cor}[maximizing worst-case growth rate]
\label{cor:rgop}
If $\mathcal{W}$ is a polyhedral subset of the probability~simplex in $\mathbb{R}^N$ representing a set of allowable portfolios and inequality $1 - \bm{w}^\intercal\bm\mu > \sqrt{\frac{(1 + (T-1)\bar\rho)\epsilon}{(1-\epsilon) T}} \cdot \allowbreak \Vert \bm\Sigma^{1/2} \bm w \Vert$ holds for every $\bm w \in \mathcal W$, then a portfolio $\bm w \in \mathcal W$ with maximum $\mathcal G_\epsilon (\bm w)$ can be obtained by solving a tractable second-order cone program whose size scales with the number of assets $N$ but is independent of the investment horizon $T$.
\end{cor}
\begin{proof}
The claim immediately follows from Theorem \ref{thm:wcvar} and so the proof is omitted.
\end{proof}

{\color{black}\subsection{Translation of autocorrelation risk to excess variance} \label{sec:risk}}
We observe that under the same first- and second-order moments of the asset return distribution, only aggregate autocorrelation $\bar\rho = (\sum_{t=1}^{T-1} \rho_t)/(T-1)$ alters $\mathcal G_\epsilon (\bm w)$. That is, individual changes in $\rho_t, ~ t = 1, \hdots, T-1$ contribute to no risk in our model as long as the aggregate correlation $\bar\rho$ remains the same. Moreover, as a consequence from Theorem \ref{thm:wcvar}, we make the following observation which encapsulates the financial risk from autocorrelations in the covariance matrix $\bm \Sigma$ of the asset return distribution. To further elaborate, we consider another representation of the result from Theorem \ref{thm:wcvar}, where we consider $-\mathcal G_\epsilon (\bm w)$ as the total risk associated with portfolio $\bm w$,
\begin{equation*}
	-\mathcal{G}_\epsilon(\bm w)\ =\ \underbrace{\vphantom{\left(\left(\sqrt{\frac{1}{T}}\right)^2\right)} \frac{1}{2\epsilon} \bm w^\intercal \bm\Sigma \bm w}_{\text{persistent risk}} \underbrace{-\ \frac{1}{2} \left( 1 - \left( 1 - \bm w^\intercal\bm\mu + \sqrt{\frac{1-\epsilon}{\epsilon T}} \Vert \hat{\bm{\Sigma}}^{1/2}\bm w \Vert\right)^2 + \frac{1}{\epsilon T}\bm{w}^\intercal\hat{\bm\Sigma} \bm w \right)}_{\text{compounding risk}},
\end{equation*}
where $\hat{\bm\Sigma}$ is a modified covariance matrix defined by $(1 + (T-1)\bar\rho)\bm\Sigma$. From this reformulation of $-\mathcal{G}_\epsilon(\bm w)$, the term $\frac{1}{2\epsilon}\bm w^\intercal \bm\Sigma \bm w$ is independent of the autocorrelations. Hence, we refer to this term as a \emph{persistent risk} and refer to the remaining part as a \emph{compounding risk}.

The persistent risk is intuitive as it is proportional to portfolio variance $\bm w^\intercal \bm\Sigma \bm w$ and inversely proportional to $\epsilon$. To understand the compounding risk better, assume that there are two investors sharing the same asset universe $\mathcal N$, the same investment horizon $\mathcal T$, and the same probabilistic preference $\epsilon$. The first investor believes that the mean and covariance matrix of the asset return distribution are given by $\bm\mu^{(1)}$ and $\bm\Sigma^{(1)}$, respectively, and her aggregate autocorrelation is $\bar\rho^{(1)}$. The second investor believes that the market is serially uncorrelated (implying that her $\bar\rho^{(2)}$ is $0$). If we further assume that the both investors share the same mean information, i.e., $\bm\mu^{(2)} = \bm\mu^{(1)}$, but the covariance matrix of the second investor is $\bm\Sigma^{(2)} = (1+(T-1)\bar\rho^{(1)})\bm\Sigma^{(1)}$. Compounding risk calculated under the view of the first investor is equivalent to that calculated under the view of the second investor. This allows us to transform a \emph{serially correlated} market into a \emph{serially uncorrelated} market by absorbing the autocorrelations in the covariance matrix.

Further to this observation, when $\bar\rho > 0$, the modified covariance matrix $\hat{\bm\Sigma} = (1+(T-1)\bar\rho)\bm\Sigma$ is larger (with respect to both non-negative and positive semidefinite cones) as $\bar\rho$ increases, implying that the investors are exposed to a higher compounding risk. Hence, they should exercise more caution as they are more exposed to potential losses. Indeed, when the autocorrelation is positive, then the market is inclined to move either upwards or downwards. By being robust, we take into account the possibility of the market moving downwards, putting more mass in the left tail of the distribution of the total profit under the worst case. {\color{black} Consequently, we anticipate that the optimal portfolio in the view of Corollary \ref{cor:rgop} favors less risky assets as $\bar\rho$ increases. We visualize this argument with an example based on real data in Section \ref{sec:asset_pref} below.} \\

{\color{black}\subsection{Asset preference in the presence of autocorrelations} \label{sec:asset_pref}}
We perform an experiment to determine a worst-case growth rate of every mean-variance efficient portfolio under different values of $\bar\rho$ ranging from 0\% to 20\% with a step size of 5\%. In this experiment, $\bm\mu$ and $\bm\Sigma$ are calibrated to the sample mean and the sample covariance matrix of \emph{``10 Industry Portfolios''} data set (from January 2003 until December 2012) provided in \emph{Fama-French}\footnote{\url{http://mba.tuck.dartmouth.edu/pages/faculty/ken.french/data_library.html}} online data library, respectively. The means and the standard deviations of the monthly returns of these 10 assets are given in Table \ref{tab:10ind}.\footnote{We exclude the covariances from Table \ref{tab:10ind} due to the lack of space.}
\begin{table}[tb]
\centering
\begin{tabular}{r|rrrrrrrrrr}
     \multicolumn{1}{c}{} & \multicolumn{10}{c}{Asset universe $\mathcal{N}$} \\
     \multicolumn{1}{c}{} & \multicolumn{1}{c}{1}    & \multicolumn{1}{c}{2}    & \multicolumn{1}{c}{3}   & \multicolumn{1}{c}{4}   & \multicolumn{1}{c}{5}   & \multicolumn{1}{c}{6}  & \multicolumn{1}{c}{7}   & \multicolumn{1}{c}{8}   & \multicolumn{1}{c}{9}    & \multicolumn{1}{c}{10}    \\ \hline
Mean (\%)  & 0.85 & 0.75 & 0.98 & 1.25 & 0.87 & 0.76 & 0.89 & 0.65 & 0.98  & 0.45  \\
Standard deviation (\%) & 3.44 & 8.53 & 5.41 & 6.19 & 5.52 & 4.66 & 4.24 & 3.66 & 3.79 & 5.62 \\ \hline \hline
\end{tabular}
\caption{Means and standard deviations of the 10 Industry Portfolios data set. \label{tab:10ind}}
\end{table}
Figure \ref{fig:wcvar_frontier} displays the worst-case growth rates of all efficient portfolios when $T$ and $\epsilon$ are set to 360 months and 20\%, respectively.
\begin{figure}
	\centering
	\begin{minipage}{6in}
  		\centering
  		$\vcenter{\hbox{\includegraphics[width=0.48\textwidth]{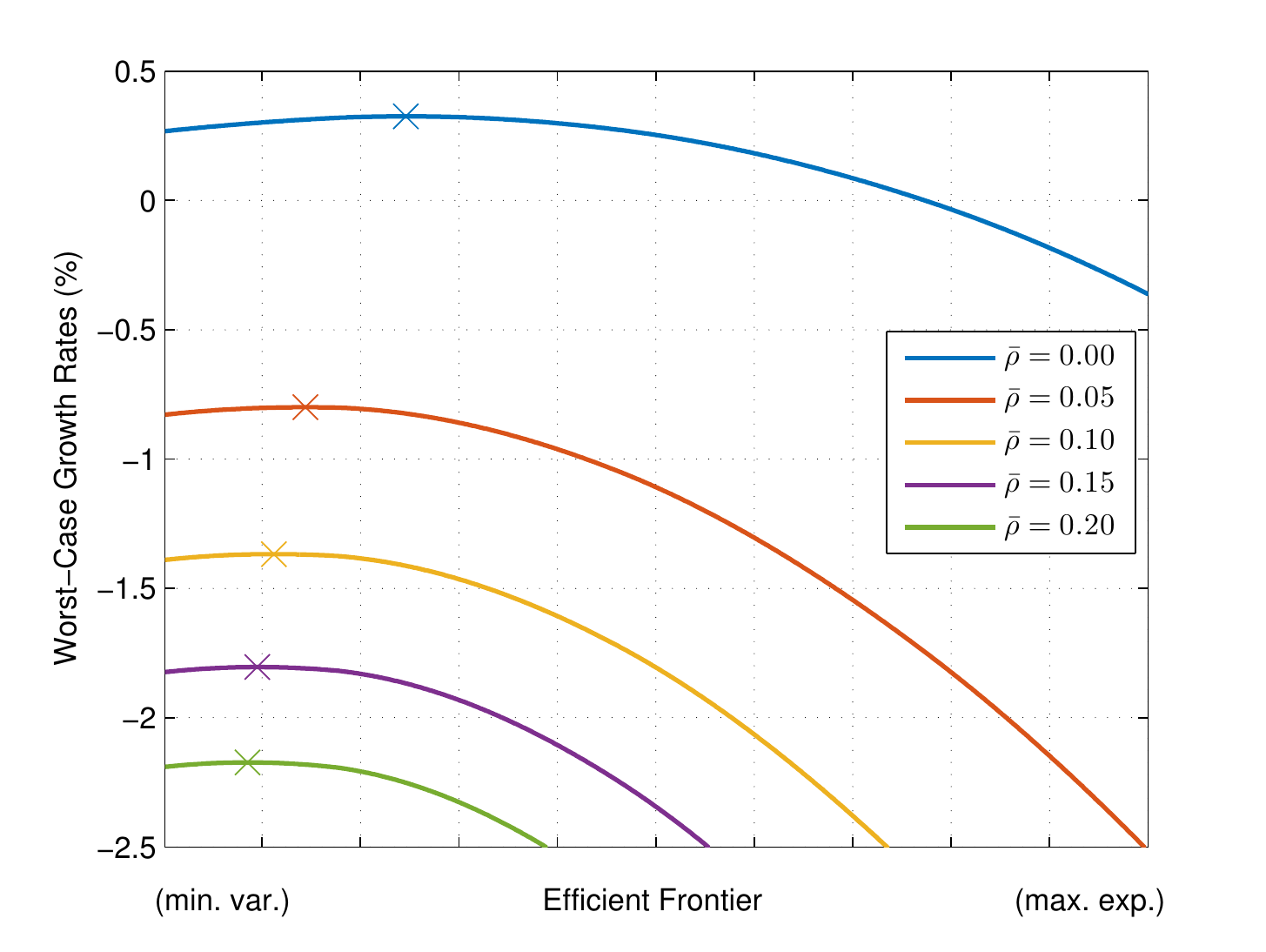}}}$
  		$\vcenter{\hbox{\vspace{-.5mm}\includegraphics[width=0.485\textwidth]{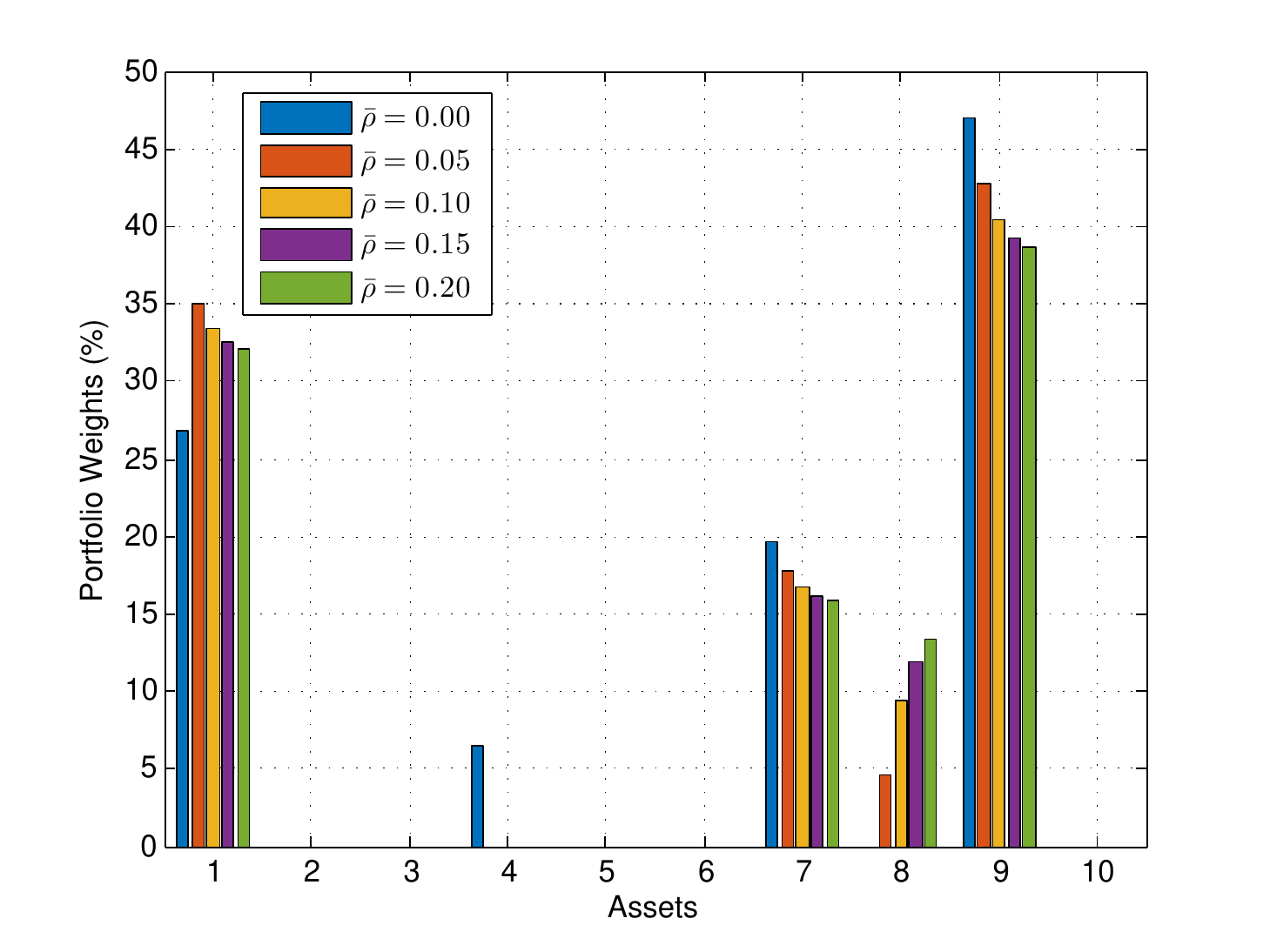}}}$
	\end{minipage}
	\caption{(Left) $\mathcal G_\epsilon (\bm w)$ calculated for every mean-variance efficient portfolio $\bm{w}$ starting from the minimum-variance portfolio on the left to maximum-expectation portfolio on the right. 
    (Right) Portfolio weights of the extended robust growth-optimal portfolios.}
	\label{fig:wcvar_frontier}
\end{figure}

It can be seen that the optimal portfolio in view of Corollary \ref{cor:rgop}, i.e., the extended robust growth-optimal portfolio, is shifted towards the minimum-variance portfolio as $\bar\rho$ increases.\footnote{Unless stated otherwise, in all experiments we assume that $\mathcal{W}$ is the probability simplex in $\mathbb{R}^N$.} Since the minimum-variance portfolio is the most risk-averse portfolio on the efficient frontier, we may say that the risk of the considered market increases with $\bar\rho$, confirming our previous hypothesis. In terms of the wealth distribution, we can see that assets with low variances (especially, assets \#1, \#7, \#8, \#9) are preferred. \\

{\color{black}\subsection{Autocorrelation Exploitation}
We now compare the performance of the two variants of the robust growth-optimal portfolios, namely the extended robust growth-optimal portfolio $\bm{w}^\text{c}$ and the original one $\bm{w}^\text{u}$. Recall that the extended robust growth-optimal portfolio is defined as an optimizer in the view of Corollary~\ref{cor:rgop}, whereas the original robust growth-optimal portfolio is ignorant of the autocorrelations and thus assumes $\bar\rho = 0$. We remark again that $\bm{w}^\text{u}$ was first studied in~\cite{Rujeerapaiboon16}. It immediately follows from our probabilistic setting that there exists a sequence of probability distributions $\left\{ \mathbb{P}^{(k)} \right\} \subset \mathcal{P}$, indexed by $k$, such that $\bm{w}^\text{c}$ outperforms $\bm{w}^\text{u}$ in the sense that 
\begin{equation*}
\begin{aligned}
	&\lim_{k \rightarrow \infty} \max_\gamma \left\lbrace \gamma :
		\mathbb{P}^\text{(k)} \left( \frac{1}{T} \sum_{t \in \mathcal T} \left( (\bm{w}^\text{c})^\intercal \tilde{\bm{r}}_t - \frac{1}{2}\left( (\bm{w}^\text{c})^\intercal \tilde{\bm{r}}_t \right)^2 \right) \geq \gamma \right) \geq 1 - \epsilon \right\rbrace \geq \\
	&\hspace{18mm} \lim_{k \rightarrow \infty} \max_\gamma \left\lbrace \gamma :
		\mathbb{P}^\text{(k)} \left( \frac{1}{T} \sum_{t \in \mathcal T} \left( (\bm{w}^\text{u})^\intercal \tilde{\bm{r}}_t - \frac{1}{2}\left((\bm{w}^\text{u})^\intercal \tilde{\bm{r}}_t \right)^2 \right) \geq \gamma \right) \geq 1 - \epsilon \right\rbrace .
\end{aligned}
\end{equation*}
However, $\mathbb{P}^\text{(k)}$ is a discrete distribution (see e.g.~\cite{Vandenberghe07, Rujeerapaiboon16}) and is therefore considered \emph{unrealistic} in financial markets. Thus, it is more practically relevant to compare $\bm{w}^\text{c}$ and $\bm{w}^\text{u}$ under a more plausible distribution.

In the following, we assume that asset returns $\tilde{\bm{r}} = \left[ \tilde{\bm{r}}^\intercal_1, \hdots , \tilde{\bm{r}}^\intercal_T \right]^\intercal$ follow a multivariate Gaussian distribution $\mathbb{P}^\text{g}$ living on $\mathbb{R}^{NT}$, with period-wise mean vector and period-wise covariance matrix calibrated to those of the 10 Industry Portfolios data set (see Table \ref{tab:10ind}). In order to save computational time, we actually restrict ourselves to the four assets favored by Corollary \ref{cor:rgop}, i.e., assets \#1, \#7, \#8, \#9 (see Figure \ref{fig:wcvar_frontier}). For two different rebalancing periods $s \neq t$, we assume constant autocorrelation $\bar\rho$ for every pair of assets (i.e., $\rho_1 = \hdots = \rho_{T-1} = \bar\rho$). Under this Gaussian assumption, we simulate 10$,$000 realizations of the random returns $\tilde{\bm{r}}$. We then compare the actual growth rates of the portfolios $\bm{w}^\text{c}$ and $\bm{w}^\text{u}$, where the actual growth rate of portfolio $\bm{w}$ is defined as an $\epsilon$--quantile of $\frac{1}{T} \sum_{t \in \mathcal T} \left( \bm{w}^\intercal \tilde{\bm{r}}_t - \frac{1}{2}\left( \bm{w}^\intercal \tilde{\bm{r}}_t \right)^2 \right)$. Figure~\ref{fig:outperformance} below presents the outperformance of the extended robust growth-optimal portfolio over its original counterpart, i.e., 
\begin{equation*}
	\text{outperformance} =  2 \cdot \frac{\left[ \widehat{\mathcal{G}}_\epsilon (\bm{w}^\text{c}) - \widehat{\mathcal{G}}_\epsilon (\bm{w}^\text{u}) \right]}{\left\vert \widehat{\mathcal{G}}_\epsilon (\bm{w}^\text{c}) \right\vert + \left\vert \widehat{\mathcal{G}}_\epsilon (\bm{w}^\text{u}) \right\vert} \quad
	(\widehat{\mathcal{G}}_\epsilon (\bm{w}) =\text{ actual growth rate of $\bm{w}$ calculated under }\mathbb{P}^\text{g})
\end{equation*}
where $\epsilon = 10\%$ and $\bar\rho = -\frac{1}{T}$.\footnote{Note that $\bar\rho$ must be between $-\frac{1}{T-1}$ and 1 for the autocorrelation matrix $\mathbf{P}$ to be positive definite.} In this experiment, we intentionally set $\bar\rho$ to be negative. In doing so, the market is considered less risky; see Section \ref{sec:risk}. Thus, we expect that the portfolio $\bm{w}^\text{u}$, which is ignorant of the autocorrelations, will be unable to compete with the portfolio $\bm{w}^\text{c}$. 
\begin{figure}
	\centering
	\begin{minipage}{6in}
  		\centering
  		$\vcenter{\hbox{\includegraphics[width=0.48\textwidth]{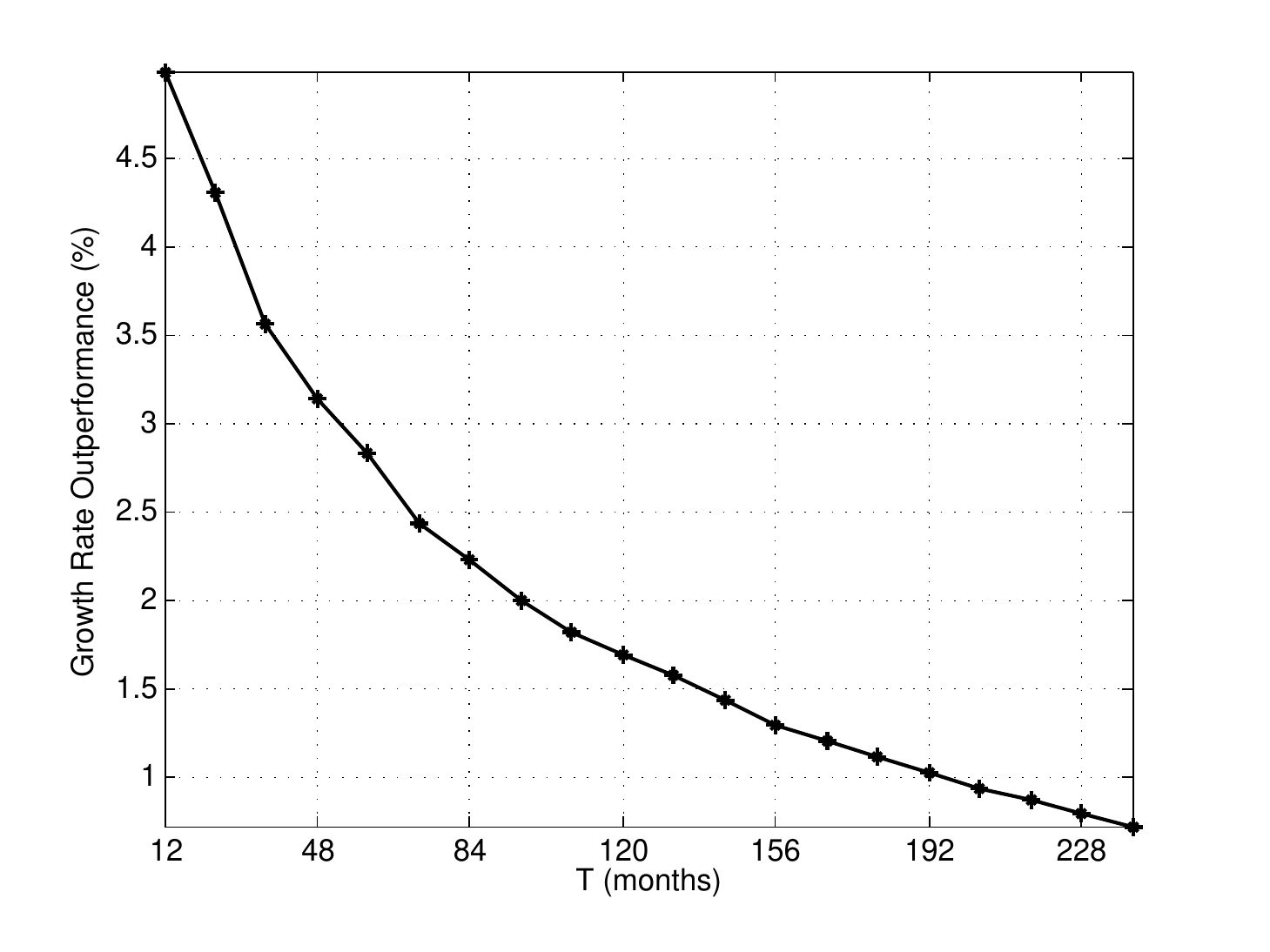}}}$
  		$\vcenter{\hbox{\vspace{-1.5mm}\includegraphics[width=0.493\textwidth]{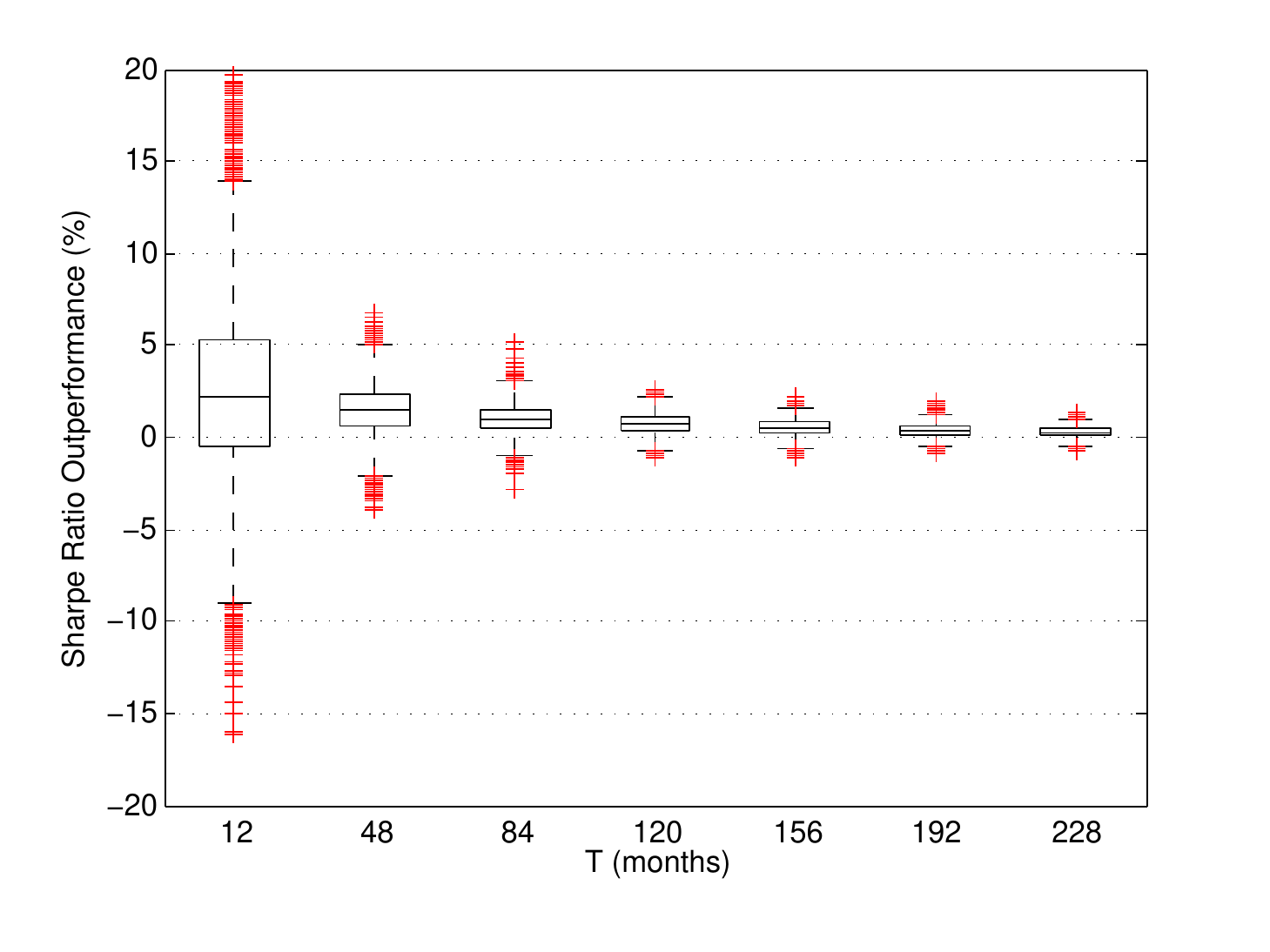}}}$
	\end{minipage}
	\caption{Comparison of the extended and the original robust growth-optimal portfolios under the Gaussian distribution $\mathbb{P}^\text{g}$ in terms of (left) actual growth rates and (right) realized Sharpe ratios.}
	\label{fig:outperformance}
\end{figure}
Indeed, Figure~\ref{fig:outperformance} confirms our intuition that the extended robust growth-optimal portfolio consistently outperforms the original version in terms of the actual growth rates across different investment horizons ranging from 1 year to 20 years. Although, the outperformance decreases with the length of the investment horizon $T$, the original robust growth-optimal portfolio becomes competitive (outperformance $<$ 1\%) only when $T > 192$ months. Under this condition though, we have that $\bar\rho > -\frac{1}{192} \approx -0.5\%$, i.e., the autocorrelations almost vanish. In addition, Figure~\ref{fig:outperformance} also contains a box plot which reports outperformance in terms of the realized Sharpe ratio which is defined as the ratio between the sample mean and the sample standard deviation of the monthly portfolio returns for each realization of $\tilde{\bm{r}}$.}

\section{Extension to a general autocorrelation matrix} \label{sec:extension}
In this section, we make an observation that the result discussed in Section \ref{sec:wcgrowthrate} can be used as a good approximation to the case with a general autocorrelation matrix. Recall that the autocorrelation matrix $\mathbf P$ defined in \eqref{eq:corrmat} is composed of at most $\lceil (T-1)/2 \rceil$ different correlation terms due to the assumption that $\rho_t = \rho_{T-t}$. Indeed, one might argue that this assumption is restrictive. If we relax this assumption, then the autocorrelations (of a weak-sense stochastic process $\tilde{\bm{\eta}}$) are expressed as a matrix of the form 
\begin{equation*}
\begin{aligned}
	\mathbf{P} = \left[ \begin{array}{cccc}
		\rho_0 & \rho_1 & \hdots & \rho_{T-1} \\
		\rho_1 & \rho_0 & \hdots & \rho_{T-2} \\
		\vdots & \vdots & \ddots & \vdots \\
		\rho_{T-1} & \rho_{T-2} & \hdots & \rho_0
	\end{array} \right]
\end{aligned}
\end{equation*}
because the correlation between $\tilde{\eta}_s$ and $\tilde{\eta}_t$ depends solely on the difference between $s$ and $t$; see e.g. \cite{Lindgren12}. In this case, $\mathbf{P}$ ceases to be circulant.

\begin{rem}
\label{rem:approx}
For a general autocorrelation matrix, we define $\bar\rho$ as $\frac{2}{T(T-1)}\sum_{t=1}^{T-1} (T-t)\rho_t$. Then, we can apply Theorem \ref{thm:wcvar} to approximate $\mathcal{G}_\epsilon (\bm{w})$ by using $\bar\rho$ as an input to the model, and we denote the approximation by $\mathcal{G}'_\epsilon (\bm{w})$. We remark that, in view of Theorem \ref{thm:wcvar}, the approximate worst-case growth rate $\mathcal{G}'_\epsilon (\bm{w})$ is obtained under an autocorrelation matrix $\mathbf{P}'$ with $\rho_t = \bar\rho$ for all $t = 1, \hdots, T-1$.
\end{rem}

We can then compare the approximate worst-case growth rate $\mathcal{G}'_\epsilon (\bm{w})$ with its exact value $\mathcal{G}_\epsilon (\bm{w})$ which is obtainable by solving the semidefinite program~\eqref{opt:projected_sdp}. The approximation errors defined by
\begin{equation*}
	\text{error} =  2 \cdot \frac{\left[ \mathcal{G}'_\epsilon (\bm{w}) - \mathcal{G}_\epsilon (\bm{w}) \right]}{\left\vert \mathcal{G}'_\epsilon (\bm{w}) \right\vert + \left\vert \mathcal{G}_\epsilon (\bm{w}) \right\vert}
\end{equation*}
are presented in Figure \ref{fig:approx} below. In this experiment, we assume that the portfolio return and standard deviation are 15\% and 20\%, respectively, and we set $\epsilon$ equal to 15\%. Note that these parameters are in line with their respective typical annual ranges (see  \cite{Luenberger98}). Furthermore, we randomly generate the correlations $\rho_t$ from a uniform distribution on the interval~$[0,0.2]$. For each value of $T$, we repeat the experiment for 20 times to get a meaningful approximation error.
\begin{figure}[!ht]
    \centering
    \includegraphics[width=1.00\textwidth]{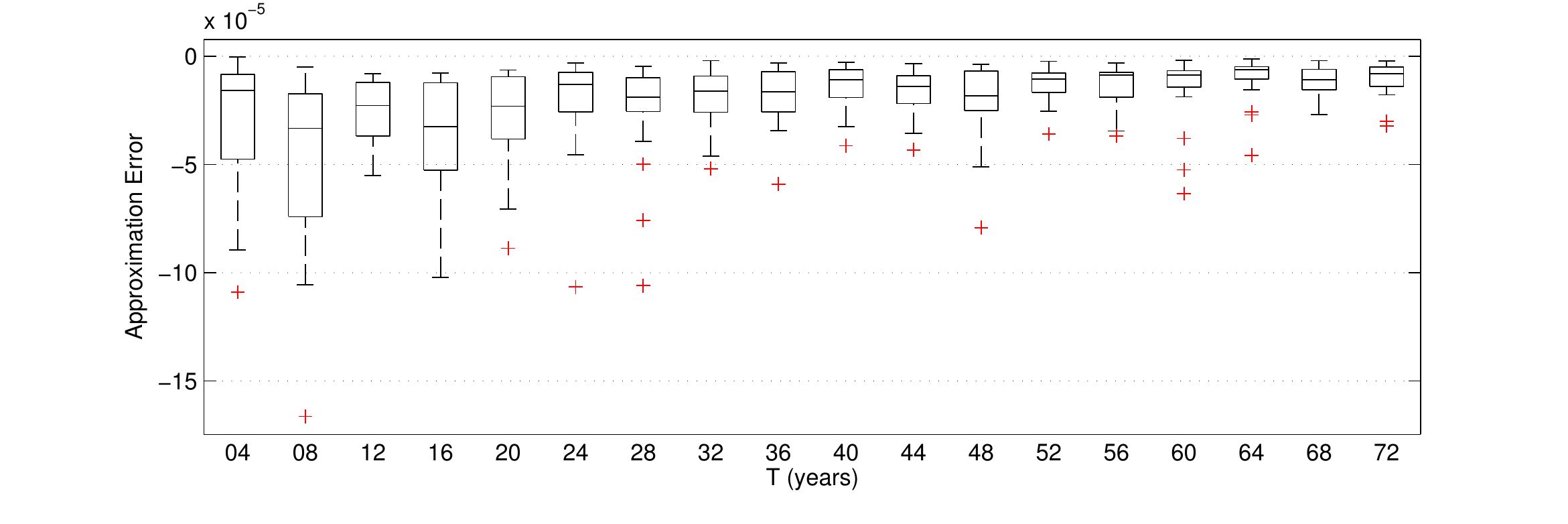}
	\caption{Approximation errors for $T = 4, 8, \hdots, 72$ (shown are 10\%, 25\%, 50\%, 75\%, 90\% quantiles and outliers).}
	\label{fig:approx}
\end{figure}
From Figure \ref{fig:approx}, it can be observed that the approximation errors are very small. Furthermore, interestingly, the errors are all negative, i.e., $\mathcal{G}_\epsilon (\bm{w}) \geq \mathcal{G}'_\epsilon (\bm{w})$. We ascertain this observation with Proposition \ref{prop:approx} below.
\begin{prop}
\label{prop:approx}
Let $\bar\rho$ denote $\frac{2}{T(T-1)}\sum_{t=1}^{T-1} (T-t)\rho_t$. 
If $1 - \bm{w}^\intercal\bm\mu > \sqrt{\frac{(1 + (T-1)\bar\rho)\epsilon}{(1-\epsilon) T}} \Vert \bm\Sigma^{1/2} \bm w \Vert$, then $\mathcal{G}_\epsilon (\bm{w}) \geq \mathcal{G}'_\epsilon (\bm{w})$.
\end{prop}
\begin{proof}
	For the proposition to hold, it suffices to prove that an optimal solution $( \mathbf{M}, \beta, \gamma)$ in the view of $\mathcal{G}'_\epsilon (\bm{w})$ is feasible in the problem determining $\mathcal{G}_\epsilon (\bm{w})$; see \eqref{opt:projected_sdp}. From the intermediate results in Section \ref{sec:wcgrowthrate}, we can assume without loss of generality that $\mathbf{M}$ is a matrix of the following form
\begin{equation*}
	\mathbf{M} = \left[ \begin{array}{cc}
		m\bm{11}^\intercal + \frac{1}{2}\mathbb{I} & m_T \bm{1} \\
		m_T \bm{1}^\intercal & m_{T+1}
	\end{array} \right] \succeq \bm{0}.
\end{equation*}
The statement of the proposition immediately follows if one can show that $\left< \mathbf{M}, \bm{\Omega}'(\bm w) - \bm{\Omega}(\bm w)\right> = 0$, where $\bm\Omega(\bm w)$ is a second-order moment matrix corresponding to the \emph{general} autocorrelation matrix $\mathbf{P}$ and $\bm\Omega'(\bm w)$ is that corresponding to the approximate autocorrelation matrix $\mathbf{P}' = (1-\bar\rho)\mathbb{I} + \bar\rho\bm{11}^\intercal$ (see Remark \ref{rem:approx}). Note that since mean and variance are stationary, 
\begin{equation*}
	 \left< \mathbf{M}, \bm{\Omega}'(\bm w) - \bm{\Omega}(\bm w)\right> = 0
	 \quad\Longleftrightarrow\quad
	 \left< m\bm{11}^\intercal + \frac{1}{2}\mathbb{I}, \mathbf{P} - \mathbf{P}' \right> = 0
	 \quad\Longleftarrow\quad
	 \left< \bm{11}^\intercal, \mathbf{P} - \mathbf{P}' \right> = 0,
\end{equation*}
where the relation on the right follows from that $\mathbf{P}$ and $\mathbf{P}'$ share the same main diagonal $\bm{1}$. The claim thus follows from the definition of $\bar\rho$ which ensures that the sum of all elements in $\mathbf{P}$ is equal to the sum of all elements in $\mathbf{P}'$ as
$\bm{1}^\intercal \mathbf{P} \bm{1} = T + 2\sum_{t=1}^{T-1} (T-t)\rho_t$ and 
$\bm{1}^\intercal \mathbf{P'} \bm{1} = T + T(T-1)\bar\rho$.
\end{proof}

{\color{black} \section{Conclusions} \label{sec:conclusion}
Inspired by \cite{Rujeerapaiboon16}, we extend the robust growth-optimal portfolio, which rigorously provides a performance guarantee tailored to an investor's investment horizon, to account for market autocorrelations. In particular, we show that if the autocorrelation matrix possesses a circulant structure, then calculating the extended robust growth-optimal portfolio is as easy as calculating its original counterpart, which assumes market uncorrelatedness. For non-circulant autocorrelation matrices, a close approximation is given. Our analysis and numerical experiments suggest that, when the aggregate autocorrelation is positive, assets with lower variances are favorable. On the other hand, accounting for negative autocorrelation may increase an investor's profitability. Last but not least, we argue that market autocorrelations can be absorbed by modifying the covariance matrix of the asset return distribution.
}

\bibliography{bibliography}

\begin{thebibliography}{10}
\expandafter\ifx\csname url\endcsname\relax
  \def\url#1{\texttt{#1}}\fi
\expandafter\ifx\csname urlprefix\endcsname\relax\def\urlprefix{URL }\fi
\expandafter\ifx\csname href\endcsname\relax
  \def\href#1#2{#2} \def\path#1{#1}\fi

\bibitem{kelly56}
J.~Kelly, A new interpretation of information rate, Bell System Technical
  Journal 35~(4) (1956) 917--926.

\bibitem{breiman61}
L.~Breiman, {Optimal gambling systems for favourable games}, in: Fourth
  Berkeley Symposium on Mathematical Statistics and Probability, University of
  California Press, 1961, pp. 65--78.

\bibitem{Algoet88}
P.~Algoet, T.~Cover, {Asymptotic optimality and asymptotic equipartition
  properties of log-optimum investment}, Annals of Probability 16~(2) (1988)
  876--898.

\bibitem{Christensen05}
M.~Christensen, On the history of the growth optimal portfolio, in: Machine
  Learning for Financial Engineering, Imperial College Press, London, 2012, pp.
  1--80.

\bibitem{MacLean10}
L.~MacLean, E.~Thorp, W.~Ziemba, Good and bad properties of the {K}elly
  criterion, in: The Kelly Capital Growth Investment Criterion: Theory and
  Practice, World Scientific, 2010, pp. 563--574.

\bibitem{Poundstone05}
W.~Poundstone, {Fortune's Formula: The Untold Story of the Scientific Betting
  System that Beat the Casinos and Wall Street}, Hill \& Wang, 2005.

\bibitem{Rujeerapaiboon16}
N.~Rujeerapaiboon, D.~Kuhn, W.~Wiesemann, Robust growth-optimal portfolios,
  Management Science 62~(7) (2016) 2090--2109.

\bibitem{Tinic84}
S.~Tinic, R.~West, Risk and return: {J}anaury vs. the rest of the year, Journal
  of Financial Economics 13~(4) (1984) 561 -- 574.

\bibitem{Lo90}
A.~Lo, A.~MacKinlay, When are contrarian profits due to stock market
  overreaction?, Review of Financial Studies 3~(2) (1990) 175--205.

\bibitem{Roy52}
A.~Roy, Safety first and the holding of assets, Econometrica 20~(3) (1952)
  431--449.

\bibitem{Yu09}
Y.~Yu, Y.~Li, D.~Schuurmans, C.~Szepesv{\'a}ri, A general projection property
  for distribution families, in: Advances in Neural Information Processing
  Systems 22, Curran Associates, Inc., 2009, pp. 2232--2240.

\bibitem{Zymler10}
S.~Zymler, D.~Kuhn, B.~Rustem, Distributionally robust joint chance constraints
  with second-order moment information, Mathematical Programming A 137~(1-2)
  (2013) 167--198.

\bibitem{Gray01}
R.~M. Gray, Toeplitz and Circulant Matrices: A Review, Now Publishers, 2006.

\bibitem{Vandenberghe07}
L.~Vandenberghe, S.~Boyd, K.~Comanor, {Generalized Chebyshev bounds via
  semidefinite programming.}, {SIAM Review} 49~(1) (2007) 52--64.

\bibitem{Lindgren12}
G.~Lindgren, Stationary Stochastic Processes: Theory and Applications, Taylor
  \& Francis, 2012.

\bibitem{Luenberger98}
D.~Luenberger, Investment Science, Oxford University Press, 1998.

\end{thebibliography}

\end{document}